\documentclass[12pt,a4paper]{article}
\usepackage[margin=1.0in]{geometry}
\usepackage[numbers]{natbib}         
\usepackage[T1]{fontenc}
\usepackage{graphicx}       
\usepackage{amsmath}        
\usepackage{array}
\usepackage{amssymb}
\usepackage{amsmath}        
\usepackage{amsfonts}       
\usepackage{amsthm}         
\usepackage{bbding}         
\usepackage{bm}             
\usepackage{graphicx}       
\usepackage{fancyvrb}       
\usepackage[numbers]{natbib}         
\usepackage[nottoc]{tocbibind} 
\usepackage{dcolumn}        
\usepackage{booktabs}       
\usepackage{paralist}       
\usepackage[usenames]{xcolor}  

\usepackage{float}
\usepackage{enumitem}
\usepackage{titling}
\usepackage{lipsum}

\newcommand{\pv}[1]{\ensuremath{pv(#1)}}

\newcommand{\chrom}[1]{\ensuremath{\chi_{p}(#1)}}

\newcommand{\umax}[1]{\ensuremath{\chi_{um}(#1)}}

\newdimen\longformulasindent
\newenvironment{longformulas}
{\global\longformulasindent=0pt
	\def\>{\global\advance\longformulasindent2em\relax\hspace{2em}}%
	\def\<{\global\advance\longformulasindent-2em\relax\hspace{-2em}}%
	\begin{array}{@{}>{\displaystyle\hspace{\longformulasindent}}l@{}}}
	{\end{array}}

\theoremstyle{plain}
\newtheorem{thm}{Theorem}
\newtheorem{lemma}[thm]{Lemma}

\newtheorem{claim}[thm]{Claim}

\theoremstyle{definition}
\newtheorem{defn}[thm]{Definition}

\newtheorem{?}[thm]{Problem}

\theoremstyle{plain}
\newtheorem{prop}{Proposition}

\author{ Jan Soukup \thanks{This research was supported by the Czech Science Foundation grant GA19-08554S}
	\\
	\normalsize{Faculty of Mathematics and Physics,}\\
	\normalsize{Charles University}\\
	\normalsize{Ke Karlovu 2027/3, 121 16 Praha 2, Czech Republic}\\
	\\
	\normalsize{E-mail: soukja2@seznam.cz}
}

\title{Improved lower bounds on parity vertex colourings of binary trees}

\date{}

\begin{document}

\maketitle

\begin{abstract}
	A vertex colouring is called a \emph{parity vertex colouring} if every path in $G$ contains an odd number of occurrences of some colour. Let $\chrom{G}$ be the minimal number of colours in a parity vertex colouring of $G$. We show that $\chrom{B^*} \ge \sqrt{d} + \frac{1}{4} \log_2(d) - \frac{1}{2}$ where $B^*$ is a subdivision of the complete binary tree $B_d$. This improves the previously known bound $\chrom{B^*} \ge \sqrt{d}$ and enhances the techniques used for proving lower bounds. We use this result to show that $\chrom{T} > \sqrt[3]{\log{n}}$ where $T$ is any binary tree with $n$ vertices. These lower bounds are also lower bounds for the conflict-free colouring. We also prove that $\chrom{G}$ is not monotone with respect to minors and determine its value for cycles. 
	
	Furthermore, we study complexity of computing the parity vertex chromatic number $\chrom{G}$. We show that checking whether a vertex colouring is a parity vertex colouring is \hbox{coNP-complete}. Then we use Courcelle's theorem to prove that the problem of checking whether $\chrom{G} \le k$ is fixed-parameter tractable with respect $k$ and the treewidth of $G$. 
\end{abstract}

\section*{Introduction}
	
	\begin{defn}
		Let $G$ be a graph. A vertex colouring $c: V(G) \rightarrow \{1,\dots, k\}$ is called a \emph{parity vertex colouring} if every path in $G$ contains an odd number of occurrences of some colour.  The \emph{parity vertex chromatic number} of $G$ (denoted by $\chrom{G}$) is the minimal number of colours in a parity vertex colouring of $G$.
	\end{defn}
	
	This colouring was independently introduced by \citet{UMax} and \citet{Borow}. \citet{UMax} introduced it as a relaxation of the conflict-free and the unique-maximum colouring, motivated by the fact that this relaxation is useful for proving lower bounds on the minimal number of colours in these other colourings. A \emph{conflict-free colouring} (in the literature also known as a conflict-free colouring of hypergraphs with respect to paths) is a colouring of vertices of $G$ such that for every path $P$ in $G$ there exists a colour used exactly once on $P$. The main application of the conflict-free colouring is in frequency assignment for cellular networks. The conflict-free colouring was studied in \cite{UMax, NPcom}. 
	
	The conflict-free colouring is a relaxation of a unique-maximum colouring. A \emph{unique-maximum colouring} (in the literature also known as a unique-maximum colouring of hypergraphs with respect to paths or alternatively as a vertex ranking) is a colouring of vertices of $G$ such that for every path $P$ in $G$ the maximum colour used on $P$ is used exactly once on $P$. We denote the minimal number of colours used in a colouring of $G$ by unique-maximum chromatic number $\umax{G}$. This colouring has many applications including sparse Cholesky factorization \cite{Cholesky} or VLSI layout \cite{VLSI}. Theoretical and algorithmic properties of this colouring were studied in many papers, see e.g. \cite{Bodlaender95, UMax, Borow}. Note that the parameter $\umax{G}$ is equivalent with the height of the minimum height elimination tree of $G$, see e.g \cite{Boundontrees} for details. It is also equivalent with the tree-depth of $G$, see \cite{NESETRIL20061022} for details.
	
	\citet{Borow} began the study of the parity vertex colouring inspired by the work on the edge variant of this problem. The study of the parity edge colouring was initiated by \citet{Edge} and continued by \citet{Edge2}. It was motivated by the fact that this colouring is closely related to the problem of deciding whether a graph embeds in the hypercube and the fact that the hypercube is one of the most popular architectures used for parallel computations \cite{hypercube}.
	
	Recently, several other related colourings were introduced. Namely vertex and edge variants of the conflict-free connection colouring and the parity connection colouring of graphs (the parity versions are sometimes called odd connection colourings). The \emph{parity vertex (edge) connection colouring} is defined as a vertex (edge) colouring of a graph such that between every pair of distinct vertices, there exists a path having an odd number of occurrences of some colour on its vertices (edges). In comparison, in parity vertex (edge) colourings every path has this property. The conflict-free versions are defined analogously. The vertex variants were studied by \citet{Vertexcon} and \citet{Vertexcon2}. The edge variants were  studied by \citet{ConflictFreeConnectionsofGraphs} and \citet{ConflictfreeconnecttionCHANG2018}.
	
	Note that every parity vertex colouring is a parity vertex connection colouring but the converse is not generally true. But in a tree, there is only one path between each pair of vertices so these colourings are the same. The similar property clearly holds for the conflict-free and the edge versions of these colourings as well.      
	
	\section*{Our work and comparison to known results}
	
	In the first section we determine the parity vertex chromatic number of cycles and we show that the parity vertex chromatic number is not monotone under the minor relation. 
	
	In the second and third section we focus on binary trees. Thus all the results in this section also hold for the parity vertex connection colouring. We improve the lower bound of \citet{UMax} on the parity vertex chromatic number of subdivisions of complete binary trees. Let $B_d$ be the complete binary tree with $d$ layers and let $B^*$ be obtained from $B_d$ by replacing edges with paths. They proved that $\chrom{B^*} \ge \sqrt{d}$. We show that $\chrom{B^*} \ge \sqrt{d} + \frac{1}{4} \log_2(d) - \frac{1}{2}$. Next, we use this bound to obtain a new bound $\chrom{T} > \sqrt[3]{\log{n}}$ where $T$ is an arbitrary binary tree on $n$ vertices. Note that these lower bounds are also lower bounds on the conflict-free chromatic number.
	
	For comparison, let us discuss the known upper bounds on the parity vertex chromatic number of trees. It is easy to see that for every tree we can always find a vertex whose removal leaves parts of at most half the size. Thus, we can colour this vertex by a unique colour and colour the remaining parts recursively. This approach yields the bound $\umax{T}\le \lfloor \log{n} \rfloor + 1$ where $T$ is a tree on $n$ vertices; see \cite{Boundontrees} for details. Thus, $\chrom{T}\le \lfloor \log{n} \rfloor + 1$ because every unique-maximum colouring is also a parity vertex colouring. There are better bounds known for some other type of trees. In particular, \citet{Gregor} showed that for a binomial tree $\mathrm{Bi}_d$ it holds that $\chrom{\mathrm{Bi}_d} \le \lceil \frac{2d+3}{3} \rceil$. Since $\mathrm{Bi}_d$ has $2^d$ vertices, it follows that $\chrom{\mathrm{Bi}_d} \le \lceil \frac{2\log(n)+3}{3} \rceil$, where $n$ is the number of vertices of $\mathrm{Bi}_d$. \citet{UMax} showed a similar bound for complete binary trees. Thus even for complete binary trees and binomial trees, there is still almost a quadratic gap between the best lower and the best upper bound on the parity vertex chromatic number.    
	
	In the fourth section we study complexity of computing the parity vertex chromatic number. We start with applying the ideas of \citet{NPcom} to prove that the problem of checking whether a given colouring is a parity vertex colouring is coNP-complete. Next we use Courcelle's meta-theorem to show that computing the parity vertex chromatic number of a graphs with bounded treewidth can be done in linear time with respect to the size of the graph when we consider the number of colours to be constant.
	
	\section*{Prerequisites}	
	Throughout the paper we consider all graphs to be finite, undirected and simple. The sets of vertices and edges of a graph $G=(V,E)$ are $V(G)$, $E(G)$ respectively. Vertex colouring of a graph $G$ is a function $f:V(G) \rightarrow \mathbb{N}$. $B_d$ denotes the complete binary tree with $d$ layers. We now define useful tools to work with the parity vertex colouring. 
		
\begin{defn}
	Let $G$ be a graph, $c: V(G) \rightarrow \{1,\dots, k\}$ be a vertex colouring of $G$ with $k$ colours and $V$ be a subset of vertices of $G$. A \emph{parity vector} of $V$, denoted by $\pv{V}$,  is an element of $\{0,1\}^k$ where the $i$-th coordinate equals the parity of the number of vertices in $V$ coloured by $i$.
	
	A non-empty path $P$ in $G$ is called a \emph{parity path} if $\pv{V(P)}$ is the zero parity vector.
\end{defn}

Note that every parity path has odd length (even number of vertices). Also, a vertex colouring is a parity vertex colouring if it does not contain any parity path.
We are working mainly with the parity vertex colourings. Thus, for brevity, by a properly coloured tree we mean a tree coloured by a parity vertex colouring unless it is said otherwise.

\section{Cycles and non-monotone properties under minors}
\citet{UMax} determined the parity vertex chromatic number of paths.
\begin{lemma} [\protect{\cite{Borow,UMax}}]\label{lemma1}
	For every $n \ge 1$, $\chrom{P_{n}} = \lfloor \log{n} \rfloor + 1$.
\end{lemma}

We extend the proof to determine the parity vertex chromatic number of cycles. 
\begin{lemma}\label{lemma2}
	For every $n\ge 3$, $\chrom{C_{n}} = \lceil \log{n} \rceil + 1$.
\end{lemma}

\begin{proof}
	To see the upper bound, we colour $n-1$  consecutive vertices on the cycle optimally as on a path (by Lemma \ref{lemma1}), and we colour the remaining vertex by a new colour. This colouring uses exactly $\lfloor \log{(n-1)} \rfloor + 1 +1 $ colours. By integrality this is equal to $ \lceil \log{n} \rceil + 1$.
	
	To prove the lower bound, consider any parity vertex colouring of $C_n$ using $k$ colours. Denote the  vertices of the cycle in the order of traversal by $v_{1}, \dots, v_{n}$. For every $i \in [n]$, denote the subpath on the first $i$ vertices by $S_{i}$. Next, for every $i \in [n-1]$, denote by $S_{i+n}$, the subpath on vertices $v_{n-i+1}, \dots, v_{n}$. Observe that all subpaths $S_{i}$ are different and that the symmetric difference of any two of them (i.e. we consider only vertices contained in exactly one of them) is also a subpath of the cycle. Thus, parity vectors of all $2n-1$ paths $S_{i}$ must be different and non-zero, otherwise there would exist a parity path in $G$. Hence $2n -1 \le 2^{k} -1$. Therefore $k\ge \lceil \log{n} \rceil + 1$.
\end{proof}

The parity vertex chromatic number is monotone with respect to subgraphs because the set of all paths in a subgraph is a subset of the set of all paths in the original graph. We show that it is not monotone with respect to minors by providing a counter-example. Recall that $B_d$ is the complete binary tree with $d$ layers.

\begin{lemma}\label{lemma3}
	$\chrom{B_{4}} = 3$ .
\end{lemma}

\begin{proof}
	The parity chromatic number of $B_4$ is at least $3$ because $B_{4}$ contains $P_{7}$ as a subgraph. And since there exists a parity vertex colouring with 3 colours, as you can see on Figure \ref{obr01:b4}, we have $\chrom{B_4} = 3$. 
\end{proof}

\begin{figure}[htbp]\centering
	\includegraphics[scale=1]{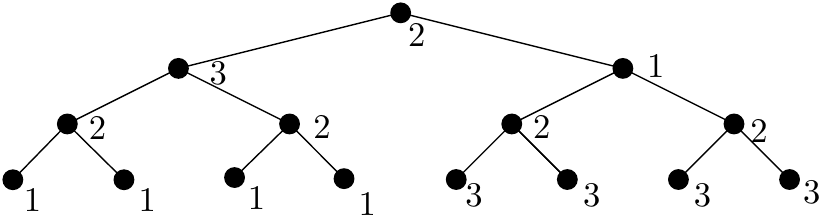}
	\caption{A parity vertex colouring of $B_{4}$ with three colours.}
	\label{obr01:b4}
\end{figure}

We denote the graph consisting of two copies of $B_{3}$ with their roots connected by $T_{3,3}$. We refer to these two rooted subtrees as the first and the second main subtree of $T_{3,3}$, see Figure \ref{obr02:t33}.
\begin{lemma} \label{lemma4}
	$\chrom{T_{3,3}} = 4.$
\end{lemma}	
\begin{proof}
	First of all, there exists a proper colouring with 4 colours, as you can see on Figure \ref{obr02:t33}. And since the tree $T_{3,3}$ contains $P_{6}$ as a subgraph, it is sufficient to prove that there is no parity vertex colouring with 3 colours.
	
	Suppose that a parity vertex colouring with $3$ colours exists.	
	At first assume that in both main subtrees there exists a leaf that has a different colour than the root of the respective subtree. In both subtrees the vertex between the respective root and the leaf has a different colour than both of them. Thus in both subtrees there exists a path starting in the root and using every colour once. But this contradicts our assumption because the roots are connected. Thus there exists a parity path, a contradiction.
	
	Now assume that one of the main subtrees, without a loss of generality the first one, has the root and all its leaves coloured with one colour. The only possibility to properly complete the parity vertex colouring of this subtree is to colour the two remaining vertices with one of the remaining colours; each with a different one. No matter what colour the root of the second subtree has, there always exists a parity path starting in this root and ending in the first subtree. Therefore our assumption was false and there is no proper colouring using 3 colours.
\end{proof}

\begin{figure}[htbp]\centering
	\includegraphics[scale=1]{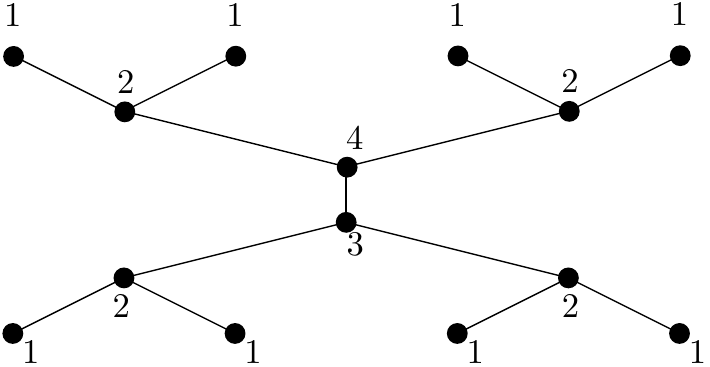}
	\caption{A parity vertex colouring of $T_{3,3}$ with four colours.}
	\label{obr02:t33}
\end{figure}

\begin{thm}
	The parity vertex chromatic number is not monotone with respect to minors.
\end{thm}
\begin{proof}
	Observe that $T_{3,3}$ is a minor of $B_4$ (it suffices to contract one of the edges incident to the root of $B_4$). By Lemma \ref{lemma3}, $\chrom{B_4}=3$ and by Lemma \ref{lemma4}, $\chrom{T_{3,3}}=4$. Therefore, a minor of $B_4$ has greater parity vertex chromatic number than $B_4$. Since there obviously exist graphs whose minors have lower parity vertex chromatic numbers, this means that the parity vertex chromatic number is not monotone with respect to minors.
\end{proof}

\section{Lower bound for subdivisions of complete binary trees}
In this section we improve the lower bound of \citet{UMax} on the parity vertex chromatic number of subdivisions of complete binary trees. Recall that by properly coloured we mean coloured by a parity vertex colouring. Say a vertex is \emph{branched} if it has at least two sons. 

\begin{defn}
	A graph $H$ is a \emph{subdivision} of a graph $G$ if $H$ is obtained from $G$ by replacing some edges with paths. The original vertices of $H$ are called the \emph{main vertices of $H$}.
\end{defn}

\begin{defn}
	Let $G$ be a properly coloured rooted binary tree such that every branched vertex, every leaf and the root of $G$ is coloured by a same colour $c$. We call the tree \emph{nicely coloured}, we call the vertices of $G$ coloured by $c$ \emph{nicely coloured} vertices and we call $c$ the nice colour of $G$.
\end{defn}

\citet{UMax} noticed that in a nicely coloured tree the parity vectors of paths connecting nicely coloured vertices with the root are distinct. Then, they found a large nicely coloured tree in every properly coloured subdivision $B^*$ of $B_d$ and showed that $\chrom{B^*} \ge \sqrt{d}$.

We generalize this idea by considering a wider class of trees to show that $\chrom{B^*} \ge \sqrt{d} + \frac{1}{4} \log(d) - \frac{1}{2}$.

\begin{defn}
	Let $T_1, T_2, \dots, T_n$ be disjoint nicely coloured trees each rooted in $r_i$. Let us connect their roots by new edges into a path $P = (r_1,\dots,r_n)$. We root the resulting tree in $r_1$ and denote it by $G$. If the colouring of $G$ (defined by the colourings of trees $T_i$) is a parity vertex colouring, then we call $G$ a \emph{safflower of trees $T_1, T_2, \dots, T_n$}, we refer to the the path $P$ as a \emph{stem} of the safflower $G$ and we refer to the trees $T_1,\dots,T_n$ as \emph{original trees} of the safflower $G$. 
	
	Additionally, we call all nicely coloured vertices of the original trees \emph{nicely coloured vertices of the safflower} and we denote the number of nicely coloured vertices of $G$ as $\mathrm{Num}(G)$.	
\end{defn}
Note that nice colours of original trees of a safflower can be distinct. 
For example, a path $(v_1, \dots, v_n)$ coloured by a parity vertex colouring is a safflower of trees $T_1, T_2, \dots, T_n$ where $T_i$ consist of just $v_i$. Its stem is the whole path. 

When we take two safflowers $S_1$ of trees $T_1, T_2, \dots, T_n$ and $S_2$ of trees $G_1, \dots, G_m$, then we can connect the root of $T_n$ and the root of $G_1$ by a new edge. And we obtain a safflower of trees $T_1, \dots, T_n, G_1, \dots, G_m$ if the colouring of the new graph is still a parity vertex colouring.
Another example of a safflower is on Figure \ref{obr03:saff}.

\begin{figure}[htbp]\centering
	\includegraphics[scale=0.8]{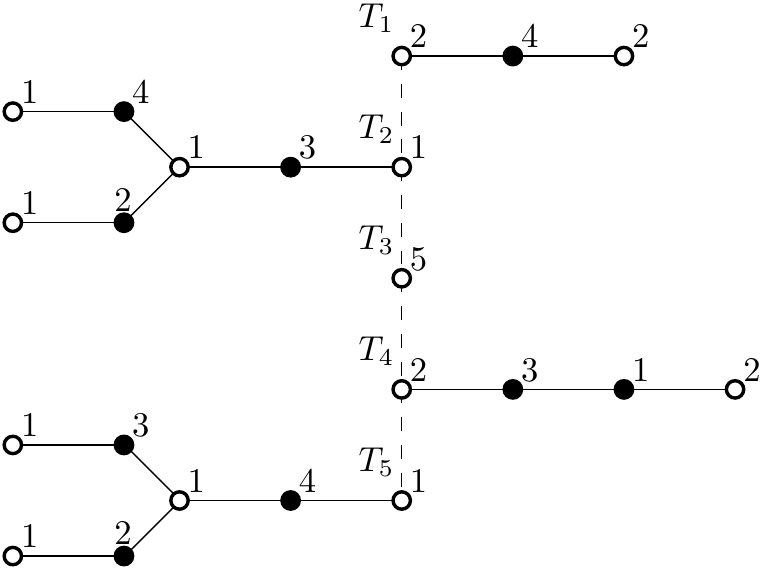}
	\caption[An example of a safflower.]{An example of a safflower of trees $T_1,\dots, T_5$. Edges of the stem of the safflower are dashed. Nicely coloured vertices of the safflower are circles, other vertices are disks.}
	\label{obr03:saff}
\end{figure}

The two following lemmas prove that every safflower uses a lot of colours. 

\begin{lemma}\label{lemma7}
	Let $F$ be a safflower with a root $r$. And let $C_1,C_2$ be paths from $r$ to arbitrary nicely coloured vertices $e_1,e_2$ of $F$, respectively. Denote the last common vertex of $C_1$ and $C_2$ by $v$. Then at least one of the vertices $e_1,e_2$ has the same colour as the vertex $v$.
\end{lemma}

\begin{proof}
	The vertex $v$ is a nicely coloured vertex of some original tree $T$ of $F$ (because it is either endpoint of one of the paths $C_1,C_2$, or it is a branched vertex of $F$). Either $v$ lies on the stem of $F$ or not. If it does not lie on the stem of $F$, then both paths $C_1,C_2$ continue into the tree $T$, and they have to end there. Since their ends are nicely coloured in $F$, they are also nicely coloured in $T$ and so they both have the same colour as $v$.
	
	If $v$ lies on the stem, then at most one of the paths continue along the stem (otherwise $v$ would not be the last common vertex). So the other path has to stay in $T$, and by the same argument as before, its endvertex has the same colour as $v$.
\end{proof}

\begin{lemma} \label{lemma7.5}
	For every safflower $F$ coloured by $k$ colours, it holds that $\mathrm{Num}(F) \le 2^k-1$. 
\end{lemma}
\begin{proof}
	Assume $F$ is coloured by $k$ colours and let $M$ be the set of all subpaths of $F$ from the root of $F$ to all nicely coloured vertices of~$F$. Clearly $|M| = \mathrm{Num}(F)$. If $\mathrm{Num}(F) =1$, then $F$ has at least one vertex, so the lemma holds. Thus assume that $|M| \ge 2$.
	
	Suppose that two distinct paths $C_1,C_2$ from $M$ have the same parity vector. Let $v$ be the last common vertex of $C_1,C_2$. Then the symmetric difference of $C_1$ and $C_2$ together with the vertex $v$ compose a subpath (denote it by $C$) of length at least 1 in $F$. Thus $\pv{C} = \pv{C_1} + \pv{C_2} + \pv{v} = \pv{v}$. By Lemma \ref{lemma7} at least one endvertex $u$ of $C$ has the same colour as $v$. Therefore, the path obtained from $C$ by deleting~$u$ is a parity path (note that it is not an empty path). This is a contradiction, and so every path in $M$ has a distinct parity vector.
	
	Furthermore, none of the paths in $M$ can have the zero parity vector. Thus, $\mathrm{Num}(F) = |M| \le 2^k -1$.
\end{proof}

Whenever we prove that every parity vertex colouring of some graph contains a safflower with some number of nicely coloured vertices as a subgraph, this lemma gives a lower bound on the parity vertex chromatic number of the graph. Since a path can be seen as a safflower that has only nicely coloured vertices, Lemma \ref{lemma7.5} generalize Lemma \ref{lemma1}. 

We show how to find a safflower with a lot of nicely coloured vertices in every subdivision of a complete binary tree coloured by a parity vertex colouring. 
We first find a nicely coloured subtrees in every properly coloured rooted binary tree. And then we inductively find a safflower of some of these nicely coloured subtrees. And finally we show that this safflower subgraph of subdivisions has a lot of nicely coloured vertices such that Lemma \ref{lemma7.5} gives the desired bound. 

Recall that for a rooted tree $T$ and its vertex $v$ we denote by $T_v$ the rooted subtree of $T$ with the root $v$ and containing exactly all the descendants of $v$ and $v$ itself. To be able to inductively construct subtrees (we want to inductively find a safflower), we need to define compatible subtrees.

\begin{defn}
	Let $T$ and $T'$ be rooted trees. We say that $T'$ is a compatible subtree of $T$ if $T'$ is a subtree of $T$ and the root of $T'$ is the closest vertex of $T'$ to the root of $T$. That is, $T'$ can be embedded in $T$ with the same descendancy relation. 
\end{defn}

Compatible subtree notation is useful for recursive definitions. Say $v$ is a vertex in a rooted binary tree $T$, with sons $l,r$. Let $S_l,S_r$ be compatible subtrees of $T_l,T_r$, respectively. Then we can connect the roots of $S_l,S_r$ with $v$ by paths in $T$ to obtain a compatible subtree of $T_v$ consisting of $S_l,S_r$ and $v$ (and the paths joining them). 

We now find some compatible nicely coloured trees in every properly coloured rooted binary tree. We will afterwards used them in finding safflowers.

\begin{defn}
	Let $T$ be a rooted binary tree properly coloured by colours in $[k]$. For every $i \in [k]$, there exists a compatible nicely coloured subtree of $T$ with nice colour $i$ that contains the maximal number of nicely coloured vertices (if there exists more such trees, pick one). We denote this subtree of $T$ as $G_i(T)$, and for brevity we denote the number of its nicely coloured vertices by $g_i(T)$ (thus $\mathrm{Num}(G_i(T)) = g_i(T)$).  
\end{defn}

\begin{figure}[htbp]\centering
	\includegraphics[scale=0.8]{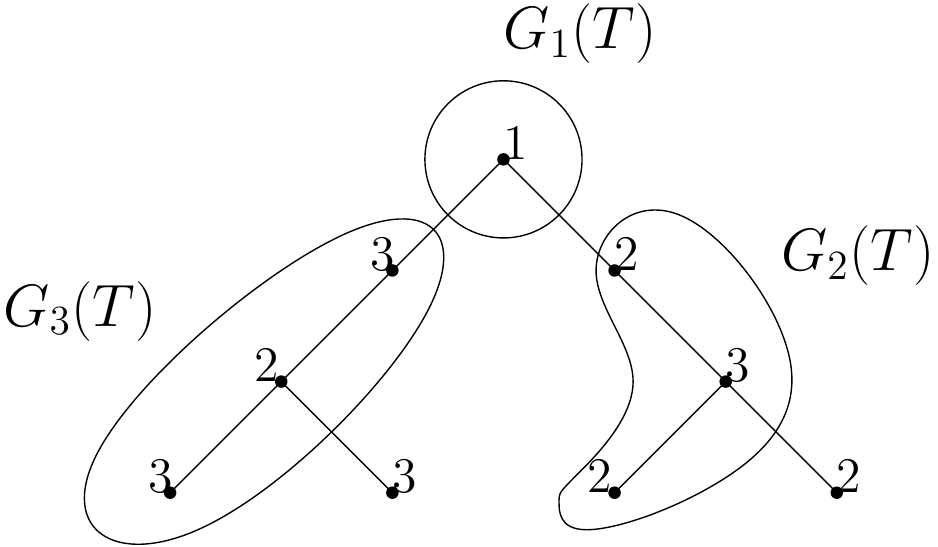}
	\caption[An example of maximal nicely coloured subtrees.]{An example of a rooted properly coloured binary tree and its maximal nicely coloured compatible subtrees. }
	\label{obr04:max_nice}
\end{figure}

See Figure \ref{obr04:max_nice} for an example. The maximality of this nicely coloured subtrees yields the following property.

\begin{lemma} \label{lemma8:maximality}
	Let $T$ be a rooted binary tree properly coloured by colours in $[k]$. Let $v$ be its vertex with sons $s,t$ and let $c$ be the colour of $v$. Then $g_c(T_v) = g_c(T_t) + g_c(T_s)+1$. 
\end{lemma}
\begin{proof}
	By maximality of $g_c(T_v)$ the graph $G_c(T_v)$ contains $v$ and the maximal nicely coloured subtrees of $T_s, T_t$, that has $g_c(T_t) + g_c(T_s)$ nicely coloured vertices by their maximality.
\end{proof}

Now we find a specific safflower subgraph in every properly coloured rooted binary tree.

\begin{defn} \label{def::saf_sub}
	Let $T$ be a rooted binary tree properly coloured by colours in $[k]$. Let $r$ be the root of $T$  and let $c$ be the colour of $r$. We define by induction a subgraph $\mathrm{Saff}(T)$ of $T$ rooted in $r$, that forms a safflower of some trees.
	\begin{enumerate}
		\item If $T$ has only one vertex $r$. Then $\mathrm{Saff}(T) = G_c(T)$ (thus it consist of one vertex $r
		$, ant it is a safflower of the tree $G_c(T)$).
		\item If $r$ has just one son $s$. Then $T_s$ is a properly coloured rooted binary tree. By induction $\mathrm{Saff}(T_s)$ is a subgraph of $T_s$ that forms a safflower of some trees $T_1, \dots, T_l$ (and $s$ is the root of this safflower and of the tree $T_1$). Let $R$ be a nicely coloured tree consisting of just $r$. We connect $R$ with $\mathrm{Saff}(T_s)$ by an edge $(r,s)$ to obtain a safflower of trees $R, T_1, \dots, T_l$ and we set this subgraph as $\mathrm{Saff}(T)$.
		\item If $r$ has two sons. Let us name them by $s,t$ such that $g_c(T_s)\ge g_c(T_t)$ (if $g_c(T_s)= g_c(T_t)$, we choose the names arbitrarily). Denote as $G$ the nicely coloured compatible subtree of $T$ rooted in $r$ and consisting of $r, G_c(T_s)$, and the path connecting $r$ with the root of $G_c(T_s)$ (It is well defined, because $G_c(T_s)$ is a compatible subtree of $T_s$, and thus of $T$ as well). Just like in the previous case,  we connect $G$ with $\mathrm{Saff}(T_t)$ by the edge joining their roots $r$ and $t$ to obtain a safflower of the tree $G$ and the original trees of $\mathrm{Saff}(T_t)$. And we set this subgraph to be $\mathrm{Saff}(T)$.
	\end{enumerate} 
	We call the subgraph $\mathrm{Saff}(T)$ \emph{a main safflower subgraph of $T$}.
\end{defn}

It is clear that the recursive definition is correct and that the main safflower subgraph of $T$ is always a safflower with a stem starting in the root of $T$ and ending in some leaf of $T$. An example of a main safflower subgraph is in Figure \ref{obr05:main_safflower}.

\begin{figure}[htbp]\centering
	\includegraphics[scale=0.8]{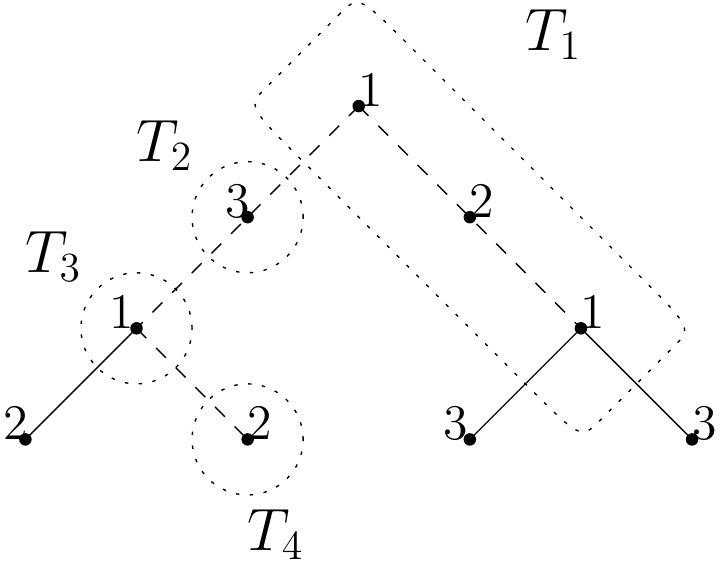}
	\caption[An example of a main safflower subgraph.]{An example of a rooted properly coloured binary tree and its main safflower subgraph (the dashed subgraph), that forms a safflower of trees $T_1, T_2, T_3, T_4$. }
	\label{obr05:main_safflower}
\end{figure}

It remains to show that the main safflower subgraph of a subdivision of a complete binary tree has lots of nicely coloured vertices. We do it terms of the two following subsequent lemmas.

\begin{lemma}\label{lemma8::orivssub}
	Let $T$ be a rooted binary tree properly coloured by colours in $[k]$, $v$ be some vertex of the stem of $\mathrm{Saff}(T)$ that is either a leaf or it is branched in $T$. Let $G$ be the original tree of $\mathrm{Saff}(T)$ rooted in $v$. Additionally, let $c$ be the colour of $v$. Then $\mathrm{Num}(G) \ge \frac{g_c(T_v)+1}{2}$.
\end{lemma}
\begin{proof}
	If $v$ is a leaf than $G$ as well as $T_v$ consist of just the vertex $v$. Thus $\mathrm{Num}(G) = g_c(T_v)=1$ and the lemma holds.
	
	Otherwise, $v$ is branched. Thus it has two sons.
	Let $s$ be the son of $v$ that is not a part of the stem of $\mathrm{Saff}(T)$ and $t$ be the other son of $v$ (thus it is a part of the stem of $\mathrm{Saff}(T)$). By Definition \ref{def::saf_sub}, they exist, it hold that $g_c(T_s)\ge g_c(T_t)$ and $G$ consists of the tree $G_c(T_s)$, the vertex $v$ and the path connecting them. Since $v$ is coloured by $c$, by Lemma \ref{lemma8:maximality}, it holds that $g_c(T_v) = 1+g_c(T_s)+g_c(T_t)$ and similarly $\mathrm{Num}(G) = g_c(T_s)+1$. Therefore,
	
	$$2\cdot \mathrm{Num}(G) = 2\left(g_c(T_s)+1\right)\ge g_c(T_s)+g_c(T_t) + 2 = g_c(T_v)+1 $$
	
	and the proof is finished.
\end{proof}

\begin{lemma} \label{lemma9:col_of_main_saf}
	Let $k,n\ge 1$ and $B^*$ be a subdivision of $B_n$ properly coloured by colours in $[k]$. Then $\mathrm{Saff}(B^*)$ has at least $\sum_{i=1}^{k} \sum_{j=0}^{a_i-1} 2^{j}$ nicely coloured vertices where $\sum_{i=1}^{k} a_i = n$ for some non-negative integers $a_i$.
\end{lemma}

\begin{proof}
	Let $r$ be a root of $B^*$. Let $S$ be the stem of $\mathrm{Saff}(B^*)$. It starts in $r$ and ends in some leaf of $B^*$. Thus it contains $n$ main vertices of $B^*$. Let us name them in the order from the leaf as $m_1, \dots, m_n$. For every $i\in[n]$ let $c(m_i)$ be the number of vertices in $\{m_j | j\le i\}$ that has the same colour as $m_i$ (thus, $c(m_i)$ is always at least 1). For every $i\in[k]$ let $K_i$ be the subset of all vertices $m_j$ that have the colour $i$. Clearly $\sum_{i=1}^{k} |K_i| = n$. Recall that every $m_i$ is a root of some original tree of $\mathrm{Saff}(B^*)$ (used in the Definition \ref{def::saf_sub}). Denote this original tree by $L_{m_i}$. Thus if we show that $\mathrm{Num}(L_{m_i}) \ge 2^{c(m_i)-1}$, then 
	
	\begin{align*}
	\mathrm{Num}(\mathrm{Saff}(B^*)) \ge \sum_{i=1}^{n} \mathrm{Num}(L_{m_i}) &= \sum_{i=1}^{k} \sum_{v\in K_i} \mathrm{Num}(L_{v}) \ge \\ &\ge \sum_{i=1}^{k} \sum_{v\in K_i} 2^{c(v)-1} 
	= \sum_{i=1}^{k} \sum_{j=1}^{|K_i|} 2^{j-1}
	\end{align*}
	where the last equality holds because for every $c\in [k]$ the set $\{c(v)|v\in K_c\}$ is equal to the set $\{1,\dots, |K_c|\}$ by the definition of $c(m_i)$. 
	
	And that is exactly what we want to prove. Therefore, it remains to prove that $\mathrm{Num}(L_{m_i}) \ge 2^{c(m_i)-1}$. Let us use the notation $T_v$ for the subgraph of $B^*$ that is rooted in $v$ and contains all descendants of $v$. We know that $m_i$ is a main vertex of $B^*$, so it is a leaf or a branched vertex. Hence, by lemma \ref{lemma8::orivssub} we know that $\mathrm{Num}(L_{m_i}) \ge \frac{g_c(T_{m_i})+1}{2}$ where $c$ is the colour of $m_i$. Thus, it suffices to prove that  $g_c(T_{m_i}) \ge 2^{c(m_i)}-1$.\\

	We will prove it by the induction on the values of $c(m_i)$. The tree $G_c(T_{m_i})$ always contains at least the vertex $m_i$, because $c$ is the colour of $m_i$. Thus it holds for $c(m_i) = 1$.
	
	If $c(m_i) > 1$, then $m_i$ is a main vertex of $B^*$ that is not a leaf. Thus it has two sons $s,t$. Let $t$ be the son that is a part of the stem of $\mathrm{Saff}(B^*)$. By the Definition \ref{def::saf_sub}, $g_c(T_s)\ge g_c(T_t)$. Since $c$ is the colour of $m_i$, by Lemma \ref{lemma8:maximality}, it holds that $g_c(T_{m_i}) = 1+g_c(T_s)+g_c(T_t)$. 
	Additionally, $\mathrm{Saff}(B^*)$ contains some main vertex $m_j$ for some $j<i$ ($m_j$ is under $m_i$) that has the same colour $c$ as the vertex $m_i$ and that satisfies $c(m_j) = c(m_i)-1$. By the induction hypothesis, $g_c(T_{m_j}) \ge 2^{c(m_j)} -1= 2^{c(m_i)-1}-1$. The tree $T_{m_j}$ is a compatible subtree of $T_t$ and thus $g_c(T_t)\ge g_c(T_{m_j})$. Combining this together we see that
	
	\begin{align*}
	g_c(T_{m_i}) = 1+g_c(T_s)+g_c(T_t) &\ge 2g_c(T_t) +1 \ge 2g_c(T_{m_j}) +1 \ge \\ &\ge 2\cdot \left(2^{c(m_i)-1}-1\right) +1 = 2^{c(m_i)}-1.
	\end{align*}
	And the lemma holds.

\end{proof}

We are almost done, the last thing we need is a technical lemma that simplifies the bound given by Lemma \ref{lemma9:col_of_main_saf} and \ref{lemma7.5}.
\begin{lemma} \label{lemma10}
	For every $k,n\ge 1$, and non-negative integers $a_i, i \in [k]$, satisfying $\sum_{i=1}^{k} a_i = n$ and $\sum_{i=1}^{k}\sum_{j=0}^{a_i-1}  \left(2^{j}\right) \le 2^k -1$ it holds that $k\ge \max(\sqrt{n}, \sqrt{n} + \frac{1}{4}\log_2(n)-\frac{1}{2})$.
\end{lemma}

\begin{proof}
	For $n=1$ it holds trivially. We first use some algebraic modifications and Jensen inequality, 
	
	\begin{align*}
	\sum_{i=1}^{k} \sum_{j=0}^{a_i-1} \left(2^{j}\right) =
	\sum_{i=1}^{k}  \left(2^{a_i}-1\right) =
	\sum_{i=1}^{k}  \left(2^{a_i}\right)  -k \ge
	k\cdot 2^{\frac{1}{k}\sum_{i=1}^{k}a_i} -k=
	k\cdot 2^{\frac{n}{k}} -k.
	\end{align*}
	Thus, we know that $k\cdot \left(2^{\frac{n}{k}} -1\right)\le 2^k-1$. Since $k \ge 1$ we see that $k\ge \frac{n}{k}$, or equivalently $k\ge \sqrt{n}$, which proves the first part. Next, we rewrite the inequality to the form $k\cdot 2^{\frac{n}{k}} \le 2^k+k-1$. We bound $k$ on the left side by $\sqrt{n}$, and $k-1$ on the right side by $2^k$ to obtain  $\sqrt{n}\cdot 2^{\frac{n}{k}} \le 2^k+2^k$. By taking logarithm of this inequality and multiplying by $k$ we obtain quadratic inequality $\frac{1}{2}k\log_2(n) +n \le k^2+k$. By solving this standard quadratic inequality in positive $k$ we obtain 
	
	$$k\ge \frac{-1+\frac{1}{2}\log_2(n)+\sqrt{\left(-1+\frac{1}{2}\log_2(n)\right)^2+4n}}{2}\ge \sqrt{n}+\frac{1}{4}\log_2(n)-\frac{1}{2},$$
	which proves the second part.
	
\end{proof}

We have everything ready to prove the main theorem.
\begin{thm}\label{theorem11}
	For every $n\ge1$ and every subdivision $B^*$ of $B_n$ it holds that $\chrom{B^*}\ge \max(\sqrt{n},\sqrt{n}+\frac{1}{4}\log_2(n)-\frac{1}{2})$.
\end{thm}

\begin{proof}
	Consider any proper colouring of $B^*$ with colours in $[k]$. By Lemma \ref{lemma9:col_of_main_saf}, $B^*$ contains a safflower with at least $\sum_{i=1}^{k} \sum_{j=0}^{a_i-1} 2^{j}$ nicely coloured vertices as a subgraph (the safflower $\mathrm{Saff}(B^*)$), where $\sum_{i=1}^{k} a_i = n$ for some non-negative integers $a_i$. Therefore by Lemma \ref{lemma7.5} it holds that $\sum_{i=1}^{k} \sum_{j=0}^{a_i-1} 2^{j} \le 2^k-1$. Thus by Lemma \ref{lemma10} this safflower, and consequently $B^*$ as well, uses at least $\max(\sqrt{n},\sqrt{n}+\frac{1}{4}\log_2(n)-\frac{1}{2})$ colours.
\end{proof}

\section{Lower bound on the parity chromatic number of binary trees}
In this section we prove the following theorem.

\begin{thm} \label{thm12}
	For every binary tree $B$ on $n$ vertices, $\chrom{B} > \sqrt[3]{\log{n}}$.
\end{thm}

We show that every binary tree either contains a long path or it contains a subdivision of a large complete binary tree. For this purpose, we estimate the maximum number of vertices a binary tree can have when it has a bounded number of layers and does not contain certain subdivisions.

\begin{defn}
	For integers $l,d \ge 0$ let $A(l,d)$ be the maximal number $n$ such that there exists a rooted binary tree on $n $ vertices with at most $l$ layers and not containing a subdivision of $B_{d+1}$ as a compatible subgraph.
\end{defn} 

We first determine value of $A(l,d)$ in edge cases and than show that it can be computed by simple recursion.

\begin{claim} \label{claim1}
	For every $l\ge0$ it holds that $A(l,0) = 0$. For every $d\ge l\ge0$ it holds that $A(l,d) = 2^l -1$.
\end{claim}

\begin{proof}
	A subdivision of $B_1$ is a single vertex, so any tree not containing a subdivision of $B_1$ as a compatible subgraph must be empty.
	
	A binary tree with at most $l$ layers can have at most $2^l -1$ vertices. The complete binary tree with $l$ layers has this number of vertices and does not contain a subdivision of $B_{d+1}$ for $d \ge l$.
\end{proof}

\begin{lemma} \label{claim3}
	For every $l > d > 0 $ it holds that $A(l,d) = A(l-1, d-1) + A(l-1,d) + 1$.
\end{lemma}

\begin{proof}
	Let $B$ be a binary tree with at most $l$ layers that does not contain a subdivision of $B_{d+1}$ as a compatible subgraph. Let $r$ be its root. Both child subtrees of $r$, let us denote them by $T_1,T_2$, have at most $l-1$ layers. They cannot both contain a subdivision of $B_{d}$ as a compatible subgraph. Otherwise it would be possible to connect their roots through $r$ to compose a subdivision of~$B_{d+1}$ as a compatible subgraph of $B$. Thus, without a loss of generality, $T_1$ and $T_2$ does not contain a subdivision of $B_{d+1}$ and $B_{d}$, respectively, as a compatible subgraph. Hence $|V(T_1)|\le A(l-1,d)$ and $|V(T_2)| \le A(l-1,d-1)$. Therefore $A(l,d) \le A(l-1, d-1) + A(l-1,d) + 1$.
	
	On the other hand, let $T_1$ and $T_2$ be trees with respectively $A(l-1,d)$ and $A(l-1,d-1)$ vertices, at most $l-1$ layers and not containing respectively $B_{d+1}$ and $B_{d}$ as compatible subgraphs. When we connect their roots by a new vertex, we obtain a new tree with at most $l$ layers that does not contain $B_{d+1}$ as a compatible subgraph. Hence 	$A(l,d) \ge A(l-1, d-1) + A(l-1,d) + 1$.
\end{proof}

This recursion has an explicit solution.
\begin{lemma} \label{lemma16}
	For every integers $l$ and $d$ such that $l > d \ge 0$ it holds that
	
	\begin{equation*} 
	A(l,d) = \sum_{i=0}^{d}\left[\left(2^i -1\right)\cdot \binom{l-i-1}{l-d-1}\right] + \binom{l}{d} -1.
	\end{equation*}
\end{lemma}

\begin{proof}
	We prove it by induction on both $l$ and $d$. First assume that $d=0$. In this case we need to prove that $A(l,0) = \binom{l}{0} -1$. This is true by Claim \ref{claim1}.
	
	Now assume that $l=d+1$. We need to prove that $A(l,l-1) = \sum_{i=0}^{l-1}\left(2^i -1 \right) + l-1$. We use induction on $l$ to do it. For $l=2$ we need to prove that $A(2,1)=2$. This holds because by Lemma \ref{claim3}, we get $A(2,1) = A(1,1) + A(1,0) +1$. And $A(1,1) + A(1,0) +1=2$ by Claim \ref{claim1}. For $l>2$ we use Lemma \ref{claim3}, the induction hypothesis, and Claim \ref{claim1} to obtain
	
	\begin{align*}
	A(l,l-1) &= A(l-1,l-1) + A(l-1,l-2) +1 \\
	& = 2^{l-1} - 1 + \sum_{i=0}^{l-2}\left(2^i -1 \right) + (l - 1) -1 + 1 
	 = \sum_{i=0}^{l-1}\left(2^i -1 \right) + l-1.
	\end{align*}

	Finally, assume that $l - 1> d >0$. By the induction hypothesis and Lemma \ref{claim3} we see that
	
	\begin{align*}
	A(l,d) &= A(l-1,d) + A(l-1,d-1) +1 \\
	& = \sum_{i=0}^{d}\left[\left(2^i -1\right)\cdot\binom{l-i-2}{l-d-2}\right] + \binom{l-1}{d} -1 + \\ 
	&\qquad  \qquad \sum_{i=0}^{d-1}\left[\left(2^i -1\right)\cdot\binom{l-i-2}{l-d-1}\right] + \binom{l-1}{d-1} -1 +1 \\
	& = \sum_{i=0}^{d-1}\left[\left(2^i -1\right)\cdot\left(\binom{l-i-2}{l-d-1}+\binom{l-i-2}{l-d-2}\right)\right]  + \left(2^d -1 \right) + \binom{l}{d} -1\\
	& = \sum_{i=0}^{d}\left[\left(2^i -1\right)\cdot\binom{l-i-1}{l-d-1}\right] + \binom{l}{d} -1.
	\end{align*}
\end{proof}

The formula provided by Lemma \ref{lemma16} does not have a simple form. Instead, we provide a simple upper bound.

\begin{lemma} \label{lemma17}
	For $l \ge d \ge 0, l>0$ it holds that $A(l,d) \le l^d$.
\end{lemma}

\begin{proof}
	The case when $l=d$ easily follows from Claim \ref{claim1}.		
	Now assume that $l > d$. We apply Lemma \ref{lemma16}. In case $d=0$ we see that $A(l,d) = 0 < l^0$. In case $d=1$ we get $A(l,d) = l$. In case $d=2$ and $d=3$ we get
	
	\begin{align*}
	A(l,2) & =\sum_{i=0}^{2}\left[\left(2^i -1\right)\cdot\binom{l-i-1}{l-2-1}\right] + \binom{l}{2} -1 
	= \binom{l-2}{l-3} + \binom{l}{2} +2< l^2\\
	A(l,3) & =\sum_{i=0}^{3}\left[\left(2^i -1\right)\cdot\binom{l-i-1}{l-3-1}\right] + \binom{l}{3} -1 = \binom{l-2}{2} + 3l -3  + \binom{l}{3} < l^3  
	\end{align*}
	where we used $l>d$. In case $d>3$ we get
	
	\begin{align*}
	A(l,d) & =\sum_{i=0}^{d}\left[\left(2^i -1\right)\cdot\binom{l-i-1}{l-d-1}\right] + \binom{l}{d} -1 \\
	& \le \left(2^d -1\right)\cdot\sum_{i=0}^{d}\binom{l-i-1}{l-d-1} + \binom{l}{d} -1\\
	& = \left(2^d -1\right)\cdot\binom{l}{l-d} + \binom{l}{d} -1 = \left(2^d\right)\cdot\binom{l}{d} -1 < l^d.
	\end{align*} 
	In the last step we used inequality $\left(2^d\right)\cdot\binom{l}{d} \le l^d$, which clearly holds for $l>d\ge 4$. In the second step we used the following well known Hockey-stick identity.
	
	\begin{prop}[(Hockey-stick identity)]
		For every $n\ge r>0$ it holds that
		$$\sum_{i=r}^{n} \binom{i}{r} = \binom{n+1}{r+1}.$$
	\end{prop}

\end{proof}

Now we are finally ready to prove Theorem \ref{thm12}.

\begin{proof}[Proof of Theorem \ref{thm12}]
	Let $n\ge1$ and $B$ be an arbitrary binary tree on $n$ vertices and let $b=\chrom{B}$. Root the tree in an arbitrary vertex. Lemma \ref{lemma1} implies that $B$ does not contain a path on $2^b$ vertices as a subgraph. Thus $B$ has at most $2^b-1$ layers. Theorem \ref{theorem11} implies that $B$ does not contain a subdivision of $B_{b^2+1}$ as a subgraph, thus particularly not as a compatible subgraph. Hence $n \le A(2^b-1,b^2)$. By Lemma \ref{lemma17} we see that $n\le \left(2^b-1\right)^{b^2} < 2^{b^3}$. Therefore $b> \sqrt[3]{\log n}$.
\end{proof}

\section{FPT of computing $\chrom{G}$ and coNP-completeness of deciding whether a colouring is a parity vertex colouring}

There have been so far no discussion of the computational complexity of the parity vertex chromatic number. On the other hand, the complexity of computing the unique maximum chromatic number of graphs is far more explored (see e.g. \cite{Bodlaender95}). Unfortunately, from complexity point of view this colouring is quite different from the parity vertex colouring. 

\citet{Bodlaender95} proved that deciding whether there exists a unique maximum colouring with less than $k$ colour is NP-complete. Thus, there exists a polynomial time algorithm deciding if a given colouring is a unique maximum colouring. \citet{NPcom} showed that the same problem for the conflict-free colouring is coNP-complete. We use a slight variation of their proof to show the same fact about the parity vertex colouring.

\begin{thm}
	Given a graph and its colouring, it is coNP-complete to decide whether the colouring is a parity vertex colouring.
\end{thm} 

\begin{proof}
	
	We need to prove that the problem is coNP-hard and that it lies in coNP. In other words we have to prove that this problem is at least as hard as any problem in coNP, and that given an appropriate certificate we can verify in polynomial time that an instance of this problem is not properly coloured. We start with the first part.
	
	We show that the complement of the Hamiltonian path problem can be reduced to our problem. That is, given a graph $G$ we construct in polynomial time a graph $G^*$ with colouring $Col$ of its vertices such that $G$ has no Hamiltonian path if and only if $Col$ is proper.
	
	\begin{sloppypar}
	Let $\{v_1, \dots, v_n\}$  be the vertex set of $G$. We define $G^*$ to consist of two isomorphic copies $G'$ and $G''$ of $G$ with vertex sets $\{v'_1,\dots,v'_n\}$ and $\{v''_1,\dots,v''_n\}$, respectively. Additionally, $G^*$ contains for every pair of vertices $v'_i, v''_i$ a path $P_i=(v'_i, v_{i,1}, \dots, v_{i,i-1}, v_{i,i+1},\dots, v_{i,n}, v''_i)$ where $v_{i,1}, \dots, v_{i,i-1}, v_{i,i+1},\dots, v_{i,n}$ are new vertices. 
	We now define the colouring $Col$. For every $i \in [n]$ we set $Col(v'_i) = Col(v''_i) = i$ and for every $n \ge i > j \ge 1$ we set $Col(v_{i,j}) = Col(v_{j,i}) = (i-1) \cdot n + j$. Observe that every colour is used exactly twice and that inner vertices of every $P_i$ are coloured with distinct colours. Furthermore, every two distinct paths $P_i,P_j$ use the same colour on exactly one pair of inner vertices, namely $v_{i,j}$ and $v_{j,i}$.
	\end{sloppypar}
	
	Let $G$ contain some Hamiltonian path, say $F = (v_1, v_2, \dots, v_n)$. It follows that $G^*$ also contains a Hamiltonian path obtained from $F$ by replacing every vertex $v_i$ by the path $P_i$ or its reverse in a way that the consecutive paths can be connected by an edge of $G'$ or $G''$ (so the end of the first path and the start of the second path lies in the same copy of $G$). Since every colour in $G^*$ is used exactly twice, it follows that this path is a parity path and so the colouring $Col$ is not proper.
	
	On the other hand, let $Col$ not be a proper parity colouring and let $F$ be a parity path of $G^*$. We show that $G^*$ and consequently $G$ contains a Hamiltonian path. Since $F$ is a parity path, every colour is used even number of times on $F$. Since every colour is used exactly twice in $G^*$, it follows that every colour is used twice or zero times on $F$. Recall that inner vertices of every path $P_i$ have different colours. Thus $F$ must contain some vertex of $G'$ or $G''$, say $v'_i$. Vertex~$v''_i$ is the only other vertex using the same colour as $v'_i$. Hence $F$ contains both $v'_i$ and $ v''_i$. Therefore it must contain an entire $P_j$ for some $j$ (subgraphs $G', G''$ are connected only by these paths). Since exactly one colour of every other $P_l$ is used also on $P_j$, it follows that $F$ contains vertices from all $P_l$. Therefore $F$ also contains all vertices of both $G'$, $G''$.
	
	Assume, for now, that $F$ is not Hamiltonian. Observe that if it does not contain all vertices of some $P_i$, then one of its end vertices must be on that path. Thus $F$ contains all paths $P_i$, except for at most two exceptions $P_k$, $P_l$. Suppose that it does not contain two different vertices of $P_k$, say $v_{k,i}$, $v_{k,j}$. At least one of indices $i$, $j$ is different from $l$, say $i$. Hence $F$ contains the entire $P_i$ and consequently also $v_{i,k}$. Therefore it must contain also $v_{k,i}$, a contradiction. We can use the same argumentation for $P_l$, and so it follows that $F$ does not contain at most one vertex of $P_k$ and at most one vertex of $P_l$. It easily follows that we can extend $F$ by these vertices and obtain a Hamiltonian path.
	
	So there always exists a Hamiltonian path in $G^*$, say the path $F$. Observe that in every Hamiltonian path of $G^*$ with end vertices in $G' \cup G''$, paths $P_i$ are part of the Hamiltonian path. So in this case we would obtain a Hamiltonian path in the graph $G$ by contracting paths $P_{i}$. We now show that we can modify $F$ to find such Hamiltonian path in $G^*$. If the end vertices of $F$ are adjacent we obtained a Hamiltonian cycle. Thus we can split it in between some pair of vertices $v'_i$, $v'_j$ and consequently find a Hamiltonian path in $G$. Otherwise, let $x,y$ be end vertices of $F$. The vertex $x$ must be adjacent to some vertex of $G' \cup G''$, say $v'_i$ (otherwise one of its neighbours would not be on the Hamiltonian path). So the Hamiltonian path looks like this $(x, \dots , v'_j, v'_i, v'_k, \dots y)$. We can reconnect it to obtain a Hamiltonian path $(v'_j, \dots , x, v'_i, v'_k, \dots y)$. In a similar way we can reconnect $y$ and its neighbour to obtain a Hamiltonian path with end vertices in $G' \cup G''$, and consequently a~Hamiltonian path in $G$.\\
	
	It remains to prove that the problem is in coNP. If the given certificate is a parity path, then we can easily verify in linear time that this path really contains every colour even number of times.
	
\end{proof}

This result suggests that computing the precise parity vertex chromatic number of a given graph is hard. We show that for graphs with bounded treewidth the problem is efficiently solvable if we consider the number of colours to be constant. We denote the treewidth of $G$ by $\mathrm{tw}(G)$.

We use the following well-known theorem.

\begin{thm}[(Courcelle's theorem \cite{Courcelle})] \label{thm:Courcle}
	Let $\varphi$ be a  CMSO\textsubscript{2} sentence and $G$ be a~graph given with its tree decomposition. There exists an algorithm deciding whether $\lfloor G\rfloor \models \varphi$ (that is, $G$ has a property $\varphi$) in time $f(|\varphi|,\mathrm{tw}(G)) \cdot n$ where $n$ is the size of $G$ , $\mathrm{tw}(G)$ is the~treewidth of $G$ and $f$ is some computable function.
\end{thm}

Notice that we use the variant with the counting monadic second order logic (CMSO\textsubscript{2}.). This logic has a non-standard predicate checking if the cardinality of a set is even.

By constructing CMSO\textsubscript{2} sentence deciding the problem and applying Courcelle's theorem we obtain the following theorem.

\begin{thm} \label{thm18::oddclouring}
	For a constant $k$ and a graph $G$ with bounded treewidth, the~problem of deciding whether the graph $G$ has a parity vertex colouring with $k$ colours is solvable in linear time with respect to the size of $G$. 
\end{thm}

\begin{proof}
	As we said we use Courcelle's theorem. Note that the requirement of a tree decomposition is not a problem because the problem of finding a tree decomposition of an input graph is known to be fixed-parameter tractable with respect to the treewidth of the input graph (concretely, it can be found in time $\mathrm{tw}(G)^{O(\mathrm{tw}(G)^3)} \cdot n$. See \cite{Decomp} for details). Therefore, it is sufficient to construct a CMSO\textsubscript{2} sentence with size dependent only on $k$ deciding the problem. We use variable $V$ for the set of vertices and $E$ for set of edges of the input graph.
	
	We will construct this sentence according to the following simple observation.
	There exists a proper colouring with $k$ colours if and only if we can partition vertices into $k$ sets according to their colours such that for every path in $G$ at least one of the partition sets has odd number of vertices in common with the path.
	
	This can be directly rewritten as a CMSO\textsubscript{2} sentence
	
	\begin{align*}
	\begin{longformulas}
	ParityColorable_k = \exists_{X_1, X_2, \dots , X_k \subseteq V} [Partition(X_1, \dots,X_k) \wedge \\
	\>\>\>\>\> \forall_{Y \subseteq V} Path(Y) \implies (Oddtimes(X_1,Y) \vee \\
	Oddtimes(X_2,Y) \vee \dots \vee Oddtimes(X_k,Y)) ]
	\end{longformulas}
	\end{align*}
	
	where $Path$, $Partition$ and $Oddtimes$ are auxiliary subformulas defined below expressing basic graph and colouring properties.
	
	\begin{align*}
	&\begin{longformulas}
	Partition(X_1, \dots,X_k) = \forall_{x\in V} [ (v\in X_1 \vee \dots \vee v \in X_k) \wedge  \\
	\>\>\> (v\notin X_1 \vee v\notin X_2) \wedge \dots \wedge (v\notin X_1 \vee v\notin X_k) \wedge\\
	(v\notin X_2 \vee v\notin X_3) \wedge \dots \wedge (v\notin X_2 \vee v\notin X_k) \wedge\\
	\dots \\
	(v\notin X_{k-1} \vee v\notin X_k)]
	\end{longformulas}
	\displaybreak[3]\\
	\\
	&\begin{longformulas}
	Path(X) = \exists_{Y\subseteq E}\ \exists_{x_1,x_2\in X} [x_1 \neq x_2 \wedge Conn(X,Y) \wedge Deg1(x_1,Y) \wedge \\
	\>\> Deg1(x_2,Y) \wedge \forall_{x \in X}( (x\neq x_1 \wedge x\neq x_2) \implies Deg2(x,Y))]
	\displaybreak[3]\\
	\\
	\>
	Conn(X,Y) = \forall_{A\subseteq X}[(\exists_{u\in X}\ u\in A \wedge \exists_{v\in X}\ v\notin A) \implies\\
	\>
	\implies (\exists_{e\in Y}\ \exists_{u,v \in X}(inc(e,u) \wedge inc(e,v) \wedge u\in A \wedge v\notin A))]
	\< \\
	\displaybreak[3]\\
	Deg1(x,Y) = \exists_{e_1\in Y}\ [ inc(e_1,x) \wedge \forall_{e\in Y}(e\neq e_1 \implies \neg inc(e,x))]
	\\	
	\\
	Deg2(x,Y) = \exists_{e_1,e_2 \in Y}\ [ inc(e_1,x) \wedge inc(e_2,x) \wedge e_1 \neq e_2 \wedge \\ 
	\> \forall_{e\in Y}((e\neq e_1 \wedge e\neq e_2) \implies \neg inc(e,x))]
	\\
	\\
	\end{longformulas}
	\displaybreak[3]\\
	&\begin{longformulas}
	Oddtimes(X,Y) = \exists_{A \subseteq X} [\neg Even(A) \wedge \\
	\>\>  \forall_{x\in X} \left((x\in A \implies x\in Y) \wedge (x\notin A \implies x\notin Y)\right)]\\
	\end{longformulas}
	\end{align*}	
	The  formula $Partition(X_1, \dots,X_k)$ expresses that the variables $X_1, \dots,X_k$ form a partition of $V$. It does so by ensuring that every vertex is in some partition set, and no vertex is in two partition sets.
	
	The formula $Path(X)$ expresses that vertices of $X$ form a path of at least 2 vertices. In particular, there must exist a subset $Y$ of edges such that one or two vertices in $X$ are vertices of exactly one edge from $Y$. The rest of the vertices must be vertices of exactly two edges from $Y$. Moreover, to really ensure that it is a path, $X$ and $Y$ must compose a connected graph. These properties are expresses by elementary subformulas $Conn(X,Y), Deg1(x,Y), Deg2(x,Y)$.
	
	The formula $Oddtimes(X,Y)$ checks if the intersection of $X$ and $Y$ has an odd size. It does so by using the predicate $Even$.
	
	Since the expanded sentence $ParityColorable_k$ has a size quadratic in $k$, Courcelle's theorem finishes the proof.
\end{proof}

\section*{Acknowledgement}
	I would like to thank Mgr. Petr Gregor, Ph.D. for many helpful discussions.

\bibliographystyle{abbrvnat}  
\renewcommand{\bibname}{Bibliography}
\bibliography{bibliography}

\end{document}